\documentclass[11pt,a4paper]{article}

\usepackage[truedimen,margin=27mm]{geometry} 

\usepackage{mathrsfs}
\usepackage{amssymb}
\usepackage{amsmath}
\usepackage{ascmac}
\usepackage{amsthm}
\usepackage[pdftex]{graphicx}     
\usepackage[pdftex]{color}     
\usepackage{booktabs}
\usepackage{setspace}
\usepackage{natbib}
\usepackage{times}    
\usepackage{comment}
\usepackage{hyperref}
\usepackage{titlesec}
\titleformat*{\section}{\large\bfseries}
\titleformat*{\subsection}{\it}

\newtheorem{theorem}{Theorem}

\newtheorem{algo}{Algorithm}

\theoremstyle{definition}

\newtheorem{remark}{Remark}

\def\Thh{{\widehat{\Theta}}}

\newcommand{\R}{\mathbb{R}}

\newcommand{\E}{\mathbb{E}}

\newcommand{\mG}{\mathcal{G}}

\newcommand{\mD}{\mathcal{D}}

\newcommand{\mS}{\mathcal{S}}

\newcommand{\mN}{\mathcal{N}}

\newcommand{\mL}{\mathcal{L}}

\newcommand{\argmin}{\operatornamewithlimits{argmin}}

\def\fh{{\widehat{f}}}

\newcommand{\given}{\,|\,}

\title{{\bf Ensemble Prediction via Covariate-dependent Stacking}\footnote{This version: \today}}
\date{}

\begin{document}
\doublespacing
\maketitle

\vspace{-2cm}
\begin{center}
{\large Tomoya Wakayama$^1$ and Shonosuke Sugasawa$^2$}

\medskip
$^1$Graduate School of Economics, The University of Tokyo\\
$^2$Faculty of Economics, Keio University
\end{center}

\vspace{0.5cm}
\begin{center}
{\large\bf Abstract}
\end{center}

This study proposes a novel approach to ensemble prediction, called ``covariate-dependent stacking'' (CDST). Unlike traditional stacking and model averaging methods, CDST allows model weights to vary flexibly as a function of covariates, thereby enhancing predictive performance in complex scenarios. We formulate the covariate-dependent weights through combinations of basis functions and estimate them via cross‐validation optimization. To analyze the theoretical properties, we establish an oracle inequality regarding the expected loss to be minimized for estimating model weights. Through comprehensive simulation studies and an application to large-scale land price prediction, we demonstrate that the CDST consistently outperforms conventional model averaging methods, particularly on datasets where base models fail to capture the underlying complexity. Our findings suggest that the CDST is especially valuable for, but not limited to, spatio-temporal prediction problems, offering a powerful tool for researchers and practitioners across a wide spectrum of data analysis fields.

\bigskip\noindent
{\bf Key words}: Cross-validation; Ensemble learning; Oracle inequality; Spatio-temporal prediction; Stacking

\vspace{0.5cm}
\section{Introduction}

Reliable data-driven prediction is of critical importance in decision-making in both the industry and government~\citep{kleinberg2015prediction} and is a major concern in the fields of statistics, machine learning, and econometrics. For data analysts, selecting the best predictive model from among a number of base models (also referred to as candidate models) remains challenging. One conventional approach to addressing this challenge is model selection, which has been extensively studied for more than 50 years, including the Akaike information criterion \citep[AIC;][]{akaike1973information}, Mallows' Cp \citep{mallows1973some}, leave-one-out cross-validation \citep{stone1974cross}, and others \citep{Schwarz1978bic, konishi1996generalised, konishi2008information}. These methods seek to choose the best model from among multiple base (candidate) models. However, the selected model can be different if the dataset is slightly altered, leading to prediction instability. Further, relying on a single model entails the risk of suboptimal performance. In situations where multiple plausible models can explain a phenomenon, model selection requires selecting only one model, which may discard valuable information available from other base models. As an alternative approach, model averaging is a popular way to spread risk across multiple models \citep{Claeskens_Hjort_2008, Benito2015, steel2020model}, thereby stabilizing and improving prediction accuracy. Specifically, model averaging assigns continuous (non-binary) weights to each model, ensuring stable prediction, even when the dataset suffers from small perturbations. Additionally, it has been shown to reduce predictive risk in regression estimation \citep{leung2006information}. In particular, model averaging can improve estimation accuracy when individual models are unstable and noisy \citep{yuan2005combining,wang2012model}. This is primarily because model averaging allows integrating various opinions without the need to assign zero weights to appropriate base models.

One stream in model averaging research is frequentist model averaging (FMA). Early works \citep{bates1969combination,granger1984improved} demonstrated the effectiveness of model averaging and established its theoretical underpinnings. Subsequently, information criterion-based model integration was explored by \cite{buckland1997model}, \cite{hansen2007least} and \cite{Claeskens_Hjort_2008}, leading to a plethora of subsequent studies on FMA \citep{wan2010least, ando2014model, liu2013heteroskedasticity, Liu2016GLS}. Other approaches include adaptive averaging by mixing~\citep{yang2000adaptive,yang2001adaptive,yuan2005combining}, plug-in averaging~\citep{LIU2015142}, and jackknife averaging~\citep{HANSEN201238}, among others. Moreover, FMA is widely explored in the field of machine learning as ensemble learning. One representative method is stacking, that is, learning a meta-model by using the predictions of multiple base models as input. The idea originated from \citet{wolpert1992stacked} within the context of neural networks and was subsequently adopted in a variety of contexts \citep{breiman1996stacked, van2007super, clydec2013bayesian, wakayama2024process}. The meta-model improves the prediction accuracy of the ensemble by optimizing the weight of each model's prediction. Ensemble learning, as well as bagging~\citep{breiman1996bagging} and boosting \citep{schapire1990strength}, has been theoretically justified \citep{mohri2018foundations}, empirically proven to improve the performance of machine learning algorithms, and widely used in data analysis competitions.

Bayesian model averaging (BMA) is another research stream of model averaging, which accounts for model uncertainty within Bayesian frameworks. Under BMA, model importance is assessed based on posterior probabilities, and weighted averages are used to make predictions. Starting with \cite{madigan1994model} and \cite{raftery1997bayesian}, numerous studies investigated the advantages and computational algorithms of BMA in detail \citep{hoeting1999bayesian,clyde2004model}. BMA has been used in a variety of statistical models, including linear regression \citep{fernandez2001benchmark}, time series \citep{STOCK2006515} and graphical models \citep{scutari2014bayesian}, among others. It has also been applied to real-world problems in fields as diverse as economics \citep{sala2004determinants}, finance \citep{avramov2002stock}, medicine \citep{yeung2005bayesian}, and ecology \citep{annest2009iterative}. Recently, Bayesian predictive synthesis (BPS), an extension of BMA, has been proposed with high prediction performance. \cite{mcalinn2019dynamic} developed a method for dynamically changing model weights over time, and \cite{cabel2022bayesian} and \cite{sugasawa2023bayesian} adopted it in spatial contexts and causal inference, respectively. Although BPS addresses time- or space-varying averaging weights, it requires Bayesian base models and is computationally expensive.

Inspired by BPS, we propose a method that extends stacking by varying weights based on covariates, including time and space. We estimate weights using cross-validation and analyze the theoretical properties of the proposed method, ensuring that it has desirable properties for large samples. Furthermore, through simulation experiments and real data analysis, we demonstrate that our method achieves superior prediction accuracy compared to conventional model averaging methods. The effectiveness of this method is particularly evident in datasets that cannot be successfully captured by a single model alone.

We discuss below several related works on covariate-varying regression structure. The idea that the regression structure changes locally on the covariate space has been studied. For instance, local regression~\citep{cleveland1988locally, fan1996local} can capture different regression structures in different regions of the covariates by employing a linear or polynomial regression model in each local region of the data space. Another important example is threshold regression~\citep{hansen2000,caner2004instrumental, Kourtellos_Stengos_Tan_2016,hansen2017}, which allows the regression structure to vary around a threshold value for a given covariate. This enables modeling singular structural changes by fitting different regression models on either side of the threshold. These methods consider hard domain splitting and employ simple regressions within each domain, whereas our method accounts for complex and smooth structural changes. In the context of variable-dependent ensembles, our approach differs from earlier studies \citep{sill2009feature, Capezza2021additive}: the former restricts the weights to be a linear combination of meta-features, while the latter constrains the weights so that they are nonnegative and sum to one. By contrast, our method imposes no constraints on the weights (thereby offering greater expressive power) and is computationally efficient, making it particularly suitable for applications involving large-scale data, such as the real estate data analyzed in Section~\ref{sec:app}.

The remainder of this paper is organized as follows. We introduce the proposed method and computational algorithms in Section~\ref{sec:method} and establish its theoretical properties in Section~\ref{sec:theory}. We then demonstrate the performance and interpretability of the proposed method through numerical experiments and case studies in Section~\ref{sec:num}. Section~\ref{sec:conclusion} concludes the article with a discussion and pointers for future research. All the technical details are deferred to the supplementary material.

\section{Methodology}\label{sec:method}

\subsection{Review: Ensemble Prediction via Stacking}\label{subsec: review}
Stacking is an ensemble learning algorithm that combines multiple models to improve predictive performance~\citep{wolpert1992stacked, breiman1996stacked, Leblanc1996combining}. Consider a collection of $J$ base models, $\{ f_1, f_2,\ldots, f_J \}$, where each $f_j:\R^d\to\R$ can be any predictor, such as linear regression, random forest~\citep{breiman2001random}, or Gaussian process regression~\citep{rasmussen2006gaussian}. Given a dataset $\mD := \{(x_i, y_i)\in \R^d\times\R \}_{i=1}^n$, the first step in stacking is to train each model $f_j$ separately and obtain a predictive function $\hat{f}_j(\cdot)$.
The next step is to determine the optimal weights for combining the predictions of the trained models. For point prediction, the stacked prediction is given by
\begin{equation}\label{model: ST}
\hat{g}(x) = \sum_{j=1}^J w_j \hat{f}_j(x),
\end{equation}
where $w_j \in \R$ is a stacking weight. To find optimal weights $\{\hat{w}_j\}_{j=1}^J$, leave-one-out cross-validation is customarily employed:
\begin{equation*}
\hat{w} = \argmin_{w_1, \ldots, w_J} \sum_{i=1}^n \left( y_i - \sum_{j=1}^J w_j \hat{f}_{j, -i}(x_i) \right)^2,
\end{equation*}
where $\hat{f}_{j, -i}(\cdot): \R^d\to \R$ denotes the predictive function of the $j$-th model when the $i$-th observation is left out of the training dataset. This optimization problem can be easily solved using the least squares estimator to obtain optimal stacking weights $\{\hat{w}_j\}_{j=1}^J$. Finally, the stacked model is reconstructed using the trained base models and the optimal stacking weights. This final model can be used to make predictions on new data points.

By combining multiple models through stacking, the ensemble can often achieve better predictive performance than any of the individual base models~\citep{van2007super, Le2017bayes}. The stacking allows the ensemble to learn how to optimally combine the strengths of each base model, potentially compensating for any individual weaknesses.

\subsection{Covariate-dependent Stacking}
As in the previous section, given a set of $J$ trained base predictions, $\fh_j(x)$ for $j=1,\ldots,J$ on $x\in \R^d$, we propose an ensemble prediction of the following form:
\begin{equation}\label{model: CDST}
\hat{g}(x)=\sum_{j=1}^J w_j(\tilde{x})\fh_j(x),
\end{equation}
where $\tilde{x}$ represents an additional covariate. $\tilde{x}$ can be either identical to or distinct from $x$, that is, $\tilde{x} = x$ or $\tilde{x} \neq x$. Note that weight $w_j(\tilde{x})$ is covariate-dependent and reduces to the standard stacking form discussed in Section~\ref{subsec: review} when $w_j(\tilde{x})$ is constant (i.e. independent of $\tilde{x}$). Henceforth, this approach will be referred to as covariate-dependent stacking (CDST). 

In \eqref{model: CDST}, we allow the covariates used to determine the ensemble weights, $\tilde{x}$, to be different from the covariates $x$, which are used to learn the base prediction models. This flexibility can be useful in practice. For instance, one may train the base models using all available covariates, while computing the weights based on more interpretable coordinates of $x$ if these are believed to drive structural changes (\textit{internal covariate setting}). Furthermore, in spatial applications~\citep{banerjee2014hierarchical, cresswikle2015statistics} such as real estate price prediction, the regional base models might be built using some covariates (e.g., lot size, crime rate, and distance from the nearest station), while macro location information such as latitude and longitude, not used in the base models, can be employed only to compute the weights (\textit{external covariate setting}). In both cases, the key idea is that selecting different covariates for weight calculation can help incorporate domain knowledge and balance interpretability with predictive performance. For brevity, we proceed with the discussion by letting $\tilde{x}=x$.

For CDST in \eqref{model: CDST}, the infinite-dimensional (functional) weights need to be determined. To make the optimization computationally feasible, we model $w_j(x)=\mu_j + E(x)^\top \gamma_j$ for $j=1,\ldots,J$, where $\mu_j \in \R$, $\gamma_j \in \R^M$ and $E:\R^d \to \R^M$ is a set of $M$ basis functions such as B-splines~\citep{deBoor1978practical} or radial basis functions~\citep{Buhmann_2003}. Then, we estimate $\Theta=(\gamma_1^{\top},\ldots,\gamma_J^{\top}, \mu_1,\ldots,\mu_J)^\top \in \R^{JM+J}$ through the penalized cross-validation:
\begin{equation}\label{optimal-weight}
\Thh={\rm argmin}_{\Theta} \sum_{i=1}^n \left[y_i-\sum_{j=1}^J \big\{\mu_j + E(x_i)^\top \gamma_j\big\}\fh_{j, -i}(x_i)\right]^2 + Q_{\lambda}(\Theta),
\end{equation}
where $Q_{\lambda}(\Theta)$ is the penalty term. This is not an unbiased estimator of risk due to the existence of the penalty term, but it is helpful in controlling the variance of the risk estimator. The detailed discussion is deferred to the supplementary material. Concerning the specific form of the penalty, if it is a fused-lasso~\citep{Tibshirani2004fused} or trend filtering~\citep{Kim2009trend} type, abrupt structural change can be captured, although we adopt a ridge-type penalty $Q_{\lambda}(\Theta)=\sum_{j=1}^J \lambda_j\gamma_j^\top \gamma_j$ with tuning parameters $\{\lambda_j\}_{j=1}^J$ for computational efficiency. Our aim is to obtain $\Thh$ and construct the optimal predictor.

Below, we note the ways to set basis functions. One approach to determining basis functions is to place them at equally spaced points in the covariate space. Specifically, let $\{c_m\}_{m=1}^M$ be a set of equally spaced points in $\R^d$. Then, the basis functions can be defined as $\phi_m(x) = \phi(x - c_m)$, where $\phi$ is a chosen basis function, such as a B-spline or a radial basis function. This approach is simple and straightforward but may not be optimal when the data points are not uniformly distributed in the covariate space.
An alternative approach, which we adopted in the later analysis, is to place the basis functions at the centers of the clusters obtained by applying the $k$-means algorithm~\citep{macqueen1967some} to the observed data points. Let $\{{x}_i\}_{i=1}^n\subset\R^d$ be the observed data points in the covariate space and $\{\tilde{c}_m\}_{m=1}^M$ be the cluster centers obtained by applying the $k$-means to $\{{x}_i\}_i$. Then, the basis functions are defined as $\phi_m(x) = \phi(x - \tilde{c}_m)$. This approach adapts the placement of the basis functions to the distribution of the observed data points and can lead to a better approximation performance, especially when the data points are not uniformly scattered in the covariate space. For the number of bases $M$, we suggest choosing it either by optimizing a prediction criterion such as cross-validation or by performing a clustering criterion, such as the elbow and silhouette methods~\citep{rousseeuw1987silhouettes}.

\subsection{Computation} \label{subsec: comp}
We propose two methods for ensemble weight estimation. First, we introduce an estimation approach based on the expectation-maximization algorithm (EM algorithm), which enables closed-form updates and achieves high computational efficiency. The EM algorithm is derived in the appendix for Gaussian regression models; however, it can be extended to Laplace regression models (corresponding to $\ell_1$ loss regression) and logistic regression models (applicable to classification) because the Laplace distribution can be expressed as a scale mixture of Gaussian distributions, and in the logistic regression, the Pólya-Gamma augmentation~\citep{Polson2013} is available.
Second, the weights in \eqref{optimal-weight} can also be optimized using a framework analogous to parameter estimation in a Gaussian generalized additive model (GAM) by employing the stable nested optimization algorithm proposed by \cite{wood2011fast}. This approach can also be extended to generalized linear regression. Note that non-Gaussian regression models may show a slight decrease in accuracy, as stated in \cite{wood2011fast}.

\section{Theoretical Validation}\label{sec:theory}

We examine the theoretical artifact of the proposed approach and discuss its theoretical properties. Specifically, by deriving an oracle inequality, we show that the predictor obtained through penalized cross-validation exhibits a small generalization gap.

\subsection{Motivation}
We have proposed a method based on the idea that flexible model integration can improve prediction accuracy, building upon the work of \cite{cabel2022bayesian}. As we expand the model expressiveness, it becomes crucial to address the bias-variance trade-off (Figure~\ref{fig: concept}). Increasing model expressiveness reduces bias---the discrepancy between the true function and the optimal predictor within the model space, also known as approximation error. However, this advantage may be offset by an increase in variance---the discrepancy between the predictor selected through cross-validation and the optimal predictor, also referred to as the generalization gap. Hereafter, we discuss this issue through an oracle inequality. The more general results and discussion are deferred to the appendix.

\begin{figure}[t]
  \begin{center}
  \includegraphics[width=0.35\linewidth ]{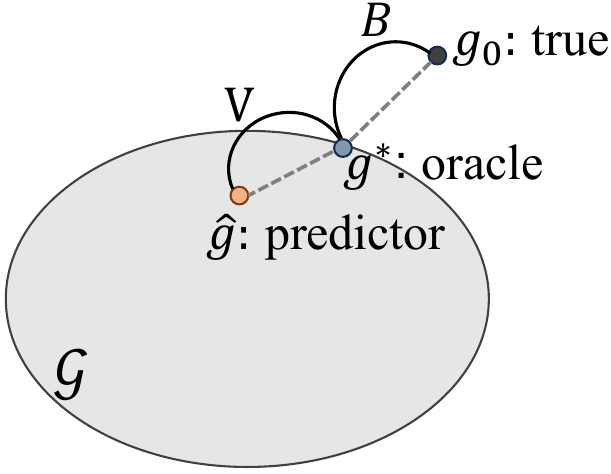}
  \end{center}
  \caption{The true data-generating function is denoted as $g_0$. The oracle predictor, $g^*$, which best approximates $g_0$ within the model combination space, $\mG$, is associated with gap $B$, known as bias or approximation error. Difference $V$ between the cross-validation-estimated $\hat{g}$ and $g^*$, known as the variance or generalization gap, is evaluated using an oracle inequality.
 \label{fig: concept}}
\end{figure}

\subsection{Setting}

Consider $S (\subset \R^d \times \R)$-valued random elements $(X_1, Y_1),\ldots,(X_n, Y_n)$ that are independent and identically distributed according to the following model:
\begin{equation*}
    Y = g_0(X) + \varepsilon~;\quad \E_P[\varepsilon\given X] =0,
\end{equation*}
where $g_0$ is a true function, $\varepsilon$ is a random error term, and $P$ is the joint distribution of $(X,Y)$. The prediction of $y$ given the new $x$ is of interest.

To closely mimic a cross-validation procedure and facilitate theoretical analysis, following the previous theories on cross-validation \citep{davies2016optimal,vaart2006oracle,van2003unified}, we adopt a randomized data-splitting setting that partitions the full dataset into a training set $\mD_0$ and a test (validation) set $\mD_1$. Formally, we assign to each observation an independent indicator random variable $I_i\in \{0,1 \}$ such that if $I_i = 0$, pair $(X_i, Y_i)$ is allocated to $\mathcal{D}_0$ and if $I_i = 1$ it is allocated to $\mathcal{D}_1$. Given the observed data $\{(X_i, Y_i)\}_{i=1}^n$, we denote by $P_I$ and $\E_I$ the probability measure and expectation operator over the random split assignments $I = (I_1, \ldots, I_n)$. The splitting is assumed independent of the data itself. In the following, $n_0$ and $n_1$ denote the cardinalities of $\mathcal{D}_0$ and $\mathcal{D}_1$, respectively, and assume $n_1 = c n$ with $0<c<1$. Let $\fh_{j}(\cdot) \ (j=1,\ldots,J)$\ be a given base predictor. We then define a stacked predictor as $g_{\Theta}(x)=\sum_{j=1}^J w_{j,\Theta}(x)\fh_{j}(x)$, where $w_{j,\Theta}(x)$ is a weight function parametrized by $\Theta$, and denote a countable class of stacked predictors $\mG = \{ g_{\Theta} \}_\Theta$. Remark that $\Theta (= \Theta(\mD_0))$ depends on $\mD_0$.

Using test data $\mathcal{D}_1$, we define the empirical risk for any stacked predictor $g_\Theta \in \mG$ as
\begin{equation*}
R_{\text{emp}}(g_\Theta; \mD_0,\mD_1) := \frac{1}{n_1}\sum_{(X_i,Y_i)\in \mD_1} L\bigl(g_{\Theta(\mD_0)}(X_i),Y_i\bigr) ,
\end{equation*}
where $L\bigl(g(X),Y\bigr)$ is a loss function given by:
\begin{equation*}
L\bigl(g(X), Y\bigr) := \bigl(Y - g(X)\bigr)^2 - \bigl(Y - g_0(X)\bigr)^2.
\end{equation*}
Note that the centering of the loss by subtracting $\bigl(Y - g_0(X)\bigr)^2$ simplifies the subsequent theoretical analysis while being equivalent to the gap between $g$ and $g_0$.
To balance the trade-off between fit and complexity, we add a measurable penalty function $Q_{\lambda}(\Theta)$ with scale parameter $\lambda$, multiplied by $1/n$. Thus, the split-averaged penalized empirical risk is defined as:
\begin{align}\label{eq:risk}
\mathcal{R}(g_\Theta; \mD_0,\mD_1) &:= \E_I\Biggl[R_{\text{emp}}(g_\Theta; \mD_0,\mD_1) + \frac{1}{n} Q_{\lambda}(\Theta)\Biggr] \nonumber \\
&= \E_I\Biggl[\frac{1}{n_1}\sum_{(X_i,Y_i)\in \mD_1} L\bigl(g_\Theta(X_i),Y_i\bigr)\Biggr] + \frac{1}{n} Q_{\lambda}(\Theta).
\end{align}
Then, we define a ``population-level" stacking predictor $\hat{g}$ as:
$$
\hat{g}(x) = g_{\hat{\Theta}}(x;\mD_0,\mD_1) := \argmin_{g_{\Theta}\in\mG} \mathcal{R}(g_\Theta; \mD_0,\mD_1)
$$
This formulation is related to equation \eqref{optimal-weight}, in that $1/n$ times the objective function in \eqref{optimal-weight} (the penalized leave-one-out cross-validation risk) is an approximation of the split-averaged penalized risk \eqref{eq:risk}. In other words, loss function $L(\cdot,\cdot)$ is used to measure the deviation of a candidate predictor from true function $g_0$, while the random splitting (with corresponding $P_I$ and $E_I$) mimics the variability inherent in cross-validation. 

Additionally, we denote the oracle predictor in $\mG$ by
\begin{align*}
    g^* &:= \argmin_{g_{\Theta}\in\mG} \E_{I} \left [ \E_{(X,Y)\sim P} \left[L(g_{\Theta(\mD_0)}(X),Y)\right]  + Q_{\lambda}(\Theta) \right] \\
    & = \argmin_{g_{\Theta}\in\mG}  \E_I \left [ \E_{\mD_1\sim P^{\otimes n_1}} \left[\frac{1}{n_1}\sum_{(X_i,Y_i)\in \mD_1} L\bigl(g_{\Theta(\mD_0)}(X_i),Y_i\bigr) \right] + \frac1n Q_{\lambda}(\Theta) \right],
\end{align*}
where $P^{\otimes n_1}$ represents the measure obtained by taking the Cartesian product of $n_1$ copies of the probability distribution $P$. This predictor is derived from the minimization of expected risk and is desirable in $\mG$. $\Theta^*$ denotes the corresponding parameter to $g^*$.

\begin{remark}[Centered loss function]
    Note that the aforementioned loss function is centered, the corresponding risk is $R(g) = \E_P[(g(X)-g_0(X))^2]$, and minimizing this risk with respect to $g$ is equivalent to minimizing the mean squared error, $\E_P[(Y-g(X))^2] = \E_P[(g(X)-g_0(X))^2] + \E_P[\varepsilon^2]$. However, for the ease of the theoretical analysis (e.g., calculation of the Bernstein pair; see the appendix for a definition), we set the loss function as a centered version.
\end{remark}

\begin{remark}[Measurability by countability]
    In the previous sections, each predictor $g_{\Theta}$ was a combination of $\{\hat{f}_j\}_{j=1}^J$ with weights $\{w
    _{j,\Theta}:\Theta\in\R^{JM+J}\}_{j=1}^J$. $\mG$ can be an uncountable set in practice, but this section assumes that it is countable for the convenience of the analysis (e.g., taking the supremum over $\mG$ preserves measurability).
\end{remark}

\subsection{Oracle Inequality}
We evaluate how stacking predictor $\hat{g}$ can estimate oracle predictor $g^*$ in the following result.

\begin{theorem} \label{thm:qdrt}
Let $S$ be a bounded and convex subset of $\R^{d+1}$ with nonempty interior. Suppose that $g_0$ and $g\in\mG$ are Lipschitz continuous, their ranges are in interval $[-\sqrt{B}/2,\sqrt{B}/2]\subset\R$ with some constant $B>0$ and $\varepsilon$ satisfies $\E_{P}[e^{t|\varepsilon|}|X]<\infty$ for some $t>0$. Then, the following inequality holds:
\begin{align*}
    &\E_I \left [ \E_{\mD_1\sim P^{\otimes n_1}} \left[\frac{1}{n_1}\sum_{(X_i,Y_i)\in \mD_1} L\bigl(\hat{g}(X_i),Y_i\bigr) \right] \right] \\ 
    & \lesssim\ \E_I \left [ \E_{\mD_1\sim P^{\otimes n_1}} \left[\frac{1}{n_1}\sum_{(X_i,Y_i)\in \mD_1} L\bigl(g^*(X_i),Y_i\bigr) \right] +\frac{Q_{\lambda}(\Theta^*)}{n}\right] 
    + \frac{\log\left(1+ n_1^{\frac{d+1}{2}} \right)}{n_1} C(t, B, g_0),
\end{align*}
where 
\begin{align*}
    C(t, B, g_0) = &\max\left\{\frac1t,1\right\} \max \{B,1\} \\
    &+ \sup_{g\in\mG} \frac{ n_1^{1/2}\E_P[(g-g_0)^2] \left(e B+ 8t^{-2}\left\|\E_{P}\left[e^{t|\varepsilon|}\given X\right] \right\|_{\infty} \right)}{ \E_P [L(g(X),Y)] },
\end{align*}
and $A_n \lesssim B_n$ means $A_n \leq cB_n$ for some constant $c > 0$ and sufficiently large $n$.
\end{theorem}

The above result measures the generalization gap of the proposed method through inequalities. The first term on the right-hand side represents the minimum penalized error. As the penalty is with a decay rate of $1/n$, the first term asymptotically converges to the minimum prediction error. The second term quantifies the discrepancy between the minimum prediction error and the prediction error of the proposed method. Recall that the essence of our stacking lies in the functional representation of the weights, and its flexibility is expressed in terms of the number $M$ of basis functions. That is, the case $M=1$ ($w$ is a constant function of $x$) corresponds to the original stacking while increasing the number of $M$ allows for adaptive combinations. Increasing the number of $M$ corresponds to growing the dimension $JM+J$ of the weights' parameter $\Theta$. Hence, the size of the predictor family $\mG$ grows exponentially, and the value of the supremum in the definition of $C(t, B, g_0)$ tends to become larger, albeit the effect is asymptotically negligible as the size of the test data increases.

\begin{remark}[Note on assumptions]
In the above theorem, the function space is assumed to be bounded, but this assumption can be replaced by the assumption that the function space (parameter set) is countably finite. The assumption on the tail property of error distribution is satisfied when the error follows a sub-Gaussian or sub-exponential distribution.
\end{remark}

\section{Numerical Studies}\label{sec:num}
We explore the implementation and predictive effectiveness of CDST through numerical experiments. We illustrate the behavior of the proposed method in Section~\ref{subsec: illustration} and compare the performance of our proposed method with other conventional methods and model averaging methods through simulations and a case study in Section~\ref{subsec: comparison}~and~\ref{sec:app}.

\subsection{Empirical behavior of CDST}\label{subsec: illustration}
To identify how the proposed method assigns weights to base models, we investigate two scenarios: one in which weights depend on covariates and the other in which they depend on spatial location in the settings of the spatial prediction. In both scenarios, we implemented the EM algorithm, the details of which are provided in the appendix. The initial values were randomly varied (from the standard normal distribution), and the results showed that the estimates for the weights and log-likelihood remained stable, with deviations within $10\%$ or $10^{-2}$, thus indicating that the estimates are not significantly influenced by the initial values. For experiments concerning the selection of the number of bases, please refer to the appendixs.

\subsection*{Case 1: Weight Dependent on Internal Covariates}
First, we consider a scenario in which the data-generating process varies depending on internal covariates; that is, observed covariates affect both the underlying regression structures and model weights. We generate $600$ covariates ${x}\in \R^5$, whose first two elements $(x_1, x_2)$ are uniformly sampled from square $[-1,1]^2\subset\R^2$ and the other three $(x_3, x_4, x_5)$ are sampled from the standard normal distribution. We consider two domains, one where $x_1$ is greater than $0$ and the remaining domain, and we set $\mu = 2 x_1 + 2 x_2$ in the former domain and $\mu = -x_1 + 4 x_2^2$ in the latter. Then, we generate $y$ by adding noise from $N(0,0.7^2)$ to $\mu$.

For our experiments, we randomly include $n=300$ points out of $600$ points for the training data. In the training data, we perform the ordinary least squares on the linear regression model for instances where the first element, $x_1$, of the covariate is less than $0$ and set it to $f_1$ in model~\eqref{model: CDST}; we similarly define $f_2$ on the other instances. We prepare $10$ radial basis function kernels $\exp (-\|(x_1, x_2)^{\top} - (c_1, c_2)^{\top}\|_2^2/2)$, where $(c_1, c_2)$ is the center of the kernel and 10 central points are selected by the $k$-means clustering~\citep{macqueen1967some, R-stats}. Then, leave-one-out cross-validation determines the weights of CDST~\eqref{model: CDST} by executing the EM algorithm until the sum of the absolute differences between the current and previous steps' parameter estimates falls below the threshold of $10^{-5}$.

Figure~\ref{fig: illustrate_cov} illustrates the estimated weights $w_1$ and $w_2$ of the model in \eqref{model: CDST} on $(x_1, x_2)\in [-1,1]^2$. In the left-hand panel, the weight $w_1$ places larger values in the region where $x_1<0$, implying that $f_1$ explains data in the region, while the right-hand panel suggests that $f_2$ accounts for the remaining area. This result is consistent with the definitions of $f_1$ and $f_2$, indicating that the weight determination performs well.

\begin{figure}[t]
  \begin{center}
  \includegraphics[width=0.9\linewidth ]{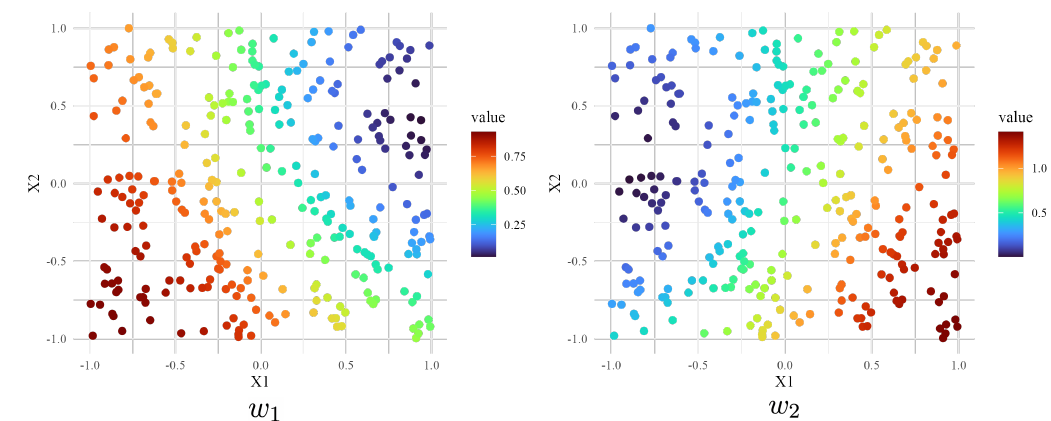}
  \end{center}
  \caption{Weights assigned by the proposed method in the scenario where the data-generating process varies depending on covariates. Left: weight for $f_1$, trained on the region where $x_1<0$. Right: weight for $f_2$, trained on the region where $x_1>0$. \label{fig: illustrate_cov}}
\end{figure}

\subsection*{Case 2: Weight Dependent on External Covariates}

We next consider a scenario where the data-generating process depends upon the model weights as a function of two-dimensional location information (i.e., external covariates), similar to \cite{cabel2022bayesian}. We uniformly sample $600$ locations ($s\in\R^2$) from square $[-1,1]^2 \subset \R^2$, whose first and second elements represent longitude and latitude, respectively. At each point, we generate covariates ${x}\in \R^5$, whose three elements $(x_3, x_4, x_5)$ are sampled from the standard normal distribution and the first two elements $(x_1, x_2)$ are sampled as follows: 
\begin{equation}\label{eq: dgp_sp}
    x_{1} = z_{1}, \qquad x_{2} = \rho z_{1} + \sqrt{1-\rho^2} z_{2},
\end{equation}
where $z_{1}$ and $z_{2}$ at all locations are sampled from $N(0, \Sigma)$ with kernel $\Sigma_{ij} = \exp(-d_{ij}/\phi)$. Here, $d_{ij}$ is the Euclidean distance between locations $s_i$ and $s_j$, $\phi = 0.5$ is the range parameter, and $\rho = 0.2$ is the correlation parameter. We examine two distinct domains: one where $x_1 > 0$ and the other corresponding to the remaining values. In the first domain, we define $\mu = x_1 - x_2^2/2$, while in the second domain, $\mu = x_1^2 + x_2^2$. The response variable $y$ is generated by adding the spatial random effect from a zero-mean Gaussian process with covariance kernel $(0.3)^2 \exp(-\|s-s'\|/0.3)$ and observation noise from $N(0, (0.7)^2)$ to $\mu$.

For our experiments, we randomly select $n=300$ points as the training dataset. With the five covariates, we run ordinary least squares on linear regression model $f_1$ for instances where the first element of $s$ is less than $0$, and we similarly perform $f_2$ for the remaining instances. We also prepare $10$ radial basis function kernels, where the central points are chosen by the $k$-means clustering from $600$ locations. Then, we compute the weights of the CDST~\eqref{model: CDST} by the EM Algorithm.

Figure~\ref{fig: illustrate_sp} displays estimated weights $w_1$ and $w_2$ of model~\eqref{model: CDST}. Large values of weight $w_1$ in the left panel are observed in the western region, meaning that $f_1$ properly explains data in the region, while the right panel suggests that $f_2$ accounts for the eastern region. Although each method alone does not account for the spatial structure, the ensemble methods are able to capture it.

\begin{figure}[t]
  \begin{center}
  \includegraphics[width=0.9\linewidth ]{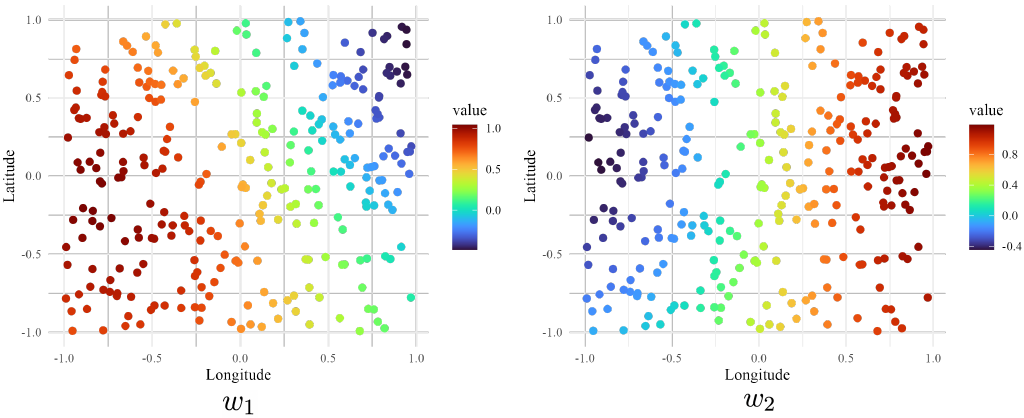}
  \end{center}
  \caption{Weights allocated by the proposed method in the scenario where the data-generating process depends on space. Left: weight for $f_1$, trained on the western region. Right: weight for $f_2$, trained on the eastern region.  \label{fig: illustrate_sp}}
\end{figure}

\subsection{Comparison of Prediction Performance}\label{subsec: comparison}
Here, we carry out Monte Carlo simulation studies to evaluate the proposed method compared to the existing methods in two cases, as considered in the previous section. As mentioned in Section~\ref{subsec: comp}, we implemented the CDST using two different algorithms: the EM algorithm (CDSTE, custom implementation) and the algorithm based on GAM (CDSTG, utilizing ``gam'' function in ``mgcv'' package in \texttt{R} language).

\subsection*{Case 1: Weight Dependent on Internal Covariates}

\begin{figure}[tb]
  \begin{center}
  \includegraphics[width=\linewidth ]{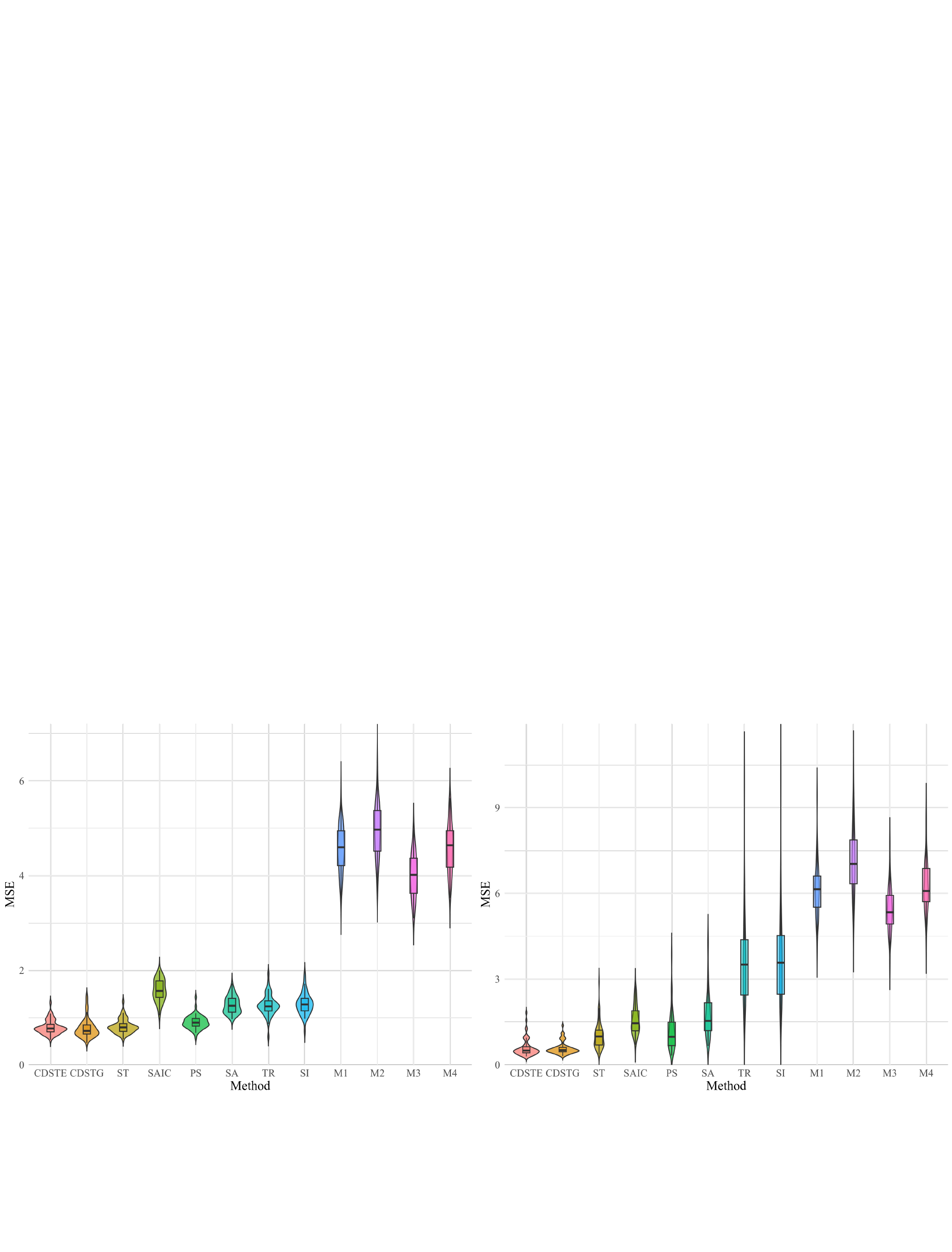}
  \end{center}
   \caption{Comparison of the six methods---linear regression (M1), additive model (M2), random forest (M3), Gaussian process regression (M4), threshold regression (TR), and single index regression (SI)---along with six ensemble predictions of (M1)--(M4)---covariate-dependent stacking by the EM algorithm (CDSTE), covariate-dependent stacking by GAM (CDSTG), probabilistic stacking (PS), vanilla stacking (ST), simple averaging (SA), and smoothed-AIC (SAIC) in each scenario. The left-hand panel corresponds to Scenario 1 and the right-hand panel to Scenario 2. \label{fig: comparison_cov}}
\end{figure}

We first examine a scenario in which the data-generating process is influenced by covariates. We generate $400$ covariates ${x}\in \R^5$, whose first two elements $(x_1, x_2)$ are uniformly sampled from square $[-1,1]^2$ and the other three $(x_3, x_4, x_5)$ are from the standard normal distribution. We prepare two scenarios, where the mean $\mu$ in each scenario is given by
\begin{itemize}
\item Scenario 1: $\mu = 2 (x_1 + x_2 + x_3) \mathbb{I}_{\{x_1<0\}}(x_1) + (4x_2^2-x_1 + x_4) \mathbb{I}_{\{x_1\ge 0\}}(x_1) $,
\item Scenario 2: $\mu = 2 (1+x_1)x_2 + (1-x_1) x_3^2$,
\end{itemize}
where $\mathbb{I}_{A}(x)$ denotes the indicator function, taking on a value of $1$ if $x\in A$, and $0$ otherwise. We then sample response variable $y$ from $N(\mu, (0.7)^2)$.
 
We repeatedly sample $100$ distinct datasets through the above procedure, and in each of them, $300$ points are utilized as training data and $100$ as test data. For CDSTE and CDSTG, we prepare the following four methods using all five covariates: (M1) linear regression by function ``lm'' in \texttt{R}; (M2) additive model with thin plate splines, as proposed by \cite{Hastie1986} and publicly available via the \texttt{R} package ``gam'' \citep{hastie2022gam}; (M3) random forest~\citep{breiman2001random} available via the \texttt{R} package “randomForest” \citep{Breiman2018R}; and (M4) Gaussian process regression \citep{rasmussen2006gaussian}, performed by \texttt{R} package “kernlab”~\citep{Karatzoglou2019kernel}. All hyperparameters are maintained at their default settings. We perform (M1)--(M4) and utilize six combinatorial methods: CDST with $10$ radial basis function kernels based on the $k$-means clustering as in Section~\ref{subsec: illustration} via the EM algorithm (CDSTE) and the optimization as GAM (CDSTG), vanilla stacking \citep[ST;][]{wolpert1992stacked} with weights optimized by leave-one-out cross-validation, probabilistic stacking~\citep[PS;][]{Capezza2021additive} using ``gamFactory'' package~\citep{gamFactory25} in \texttt{R}, simple averaging (SA) with equal weights, and smoothed-AIC \citep[SAIC;][]{buckland1997model} with AIC values of (M1) and (M2). Additionally, we implement two conventional methods in econometrics: threshold regression \citep[TR;][]{hansen2000} with a threshold set to $0$ and semiparametric single index regression using Ichimura method \citep[SI;][]{ICHIMURA199371, npPackage}. To measure prediction performance, we use the mean squared error (MSE) computed on the test datasets.

Figure~\ref{fig: comparison_cov} provides comparisons of the six single methods and six ensemble predictions. The violin plots are calculated based on $100$ datasets. CDSTE and CDSTG perform better than other ensemble methods. In scenario 1, they are almost the same as ST, yet in scenario 2, they clearly outperform ST. PS performs well, but not as well as the other stacking methods. This is because the constraint on the weights (non-negative and summing to one) limits its expressiveness, and in fact, it has been observed that CDST sometimes assigns weights that are negative or greater than one. (For plots of the weights, see the appendix.) We note that each of the methods used in the ensemble is either a simple method or a completely nonparametric method, and their individual performance is significantly inferior to methods such as TR with a known threshold and SI. However, when properly integrated, the proposed approach considerably surpasses them. 

\begin{figure}[tb]
  \begin{center}
  \includegraphics[width=\linewidth]{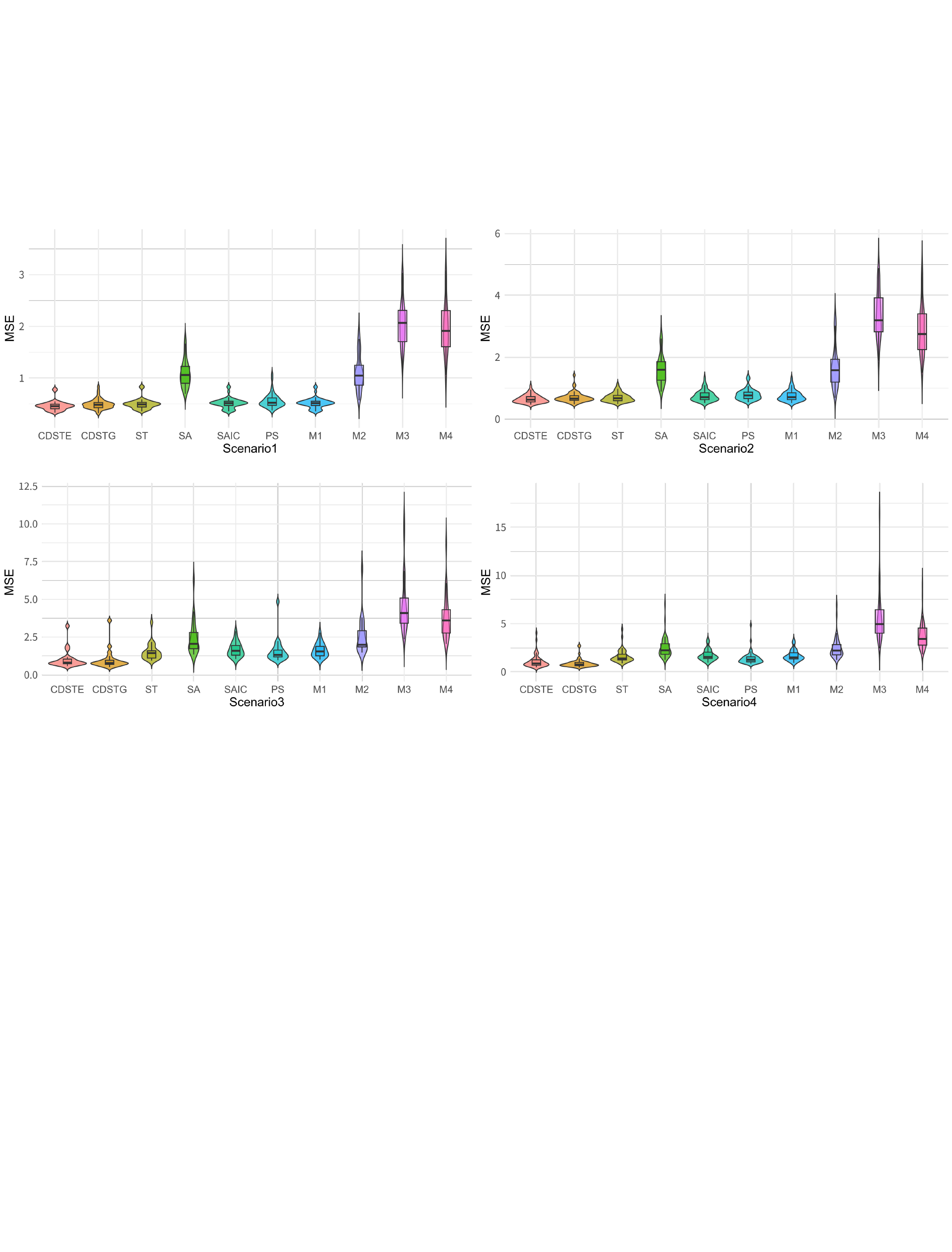}
  \end{center}
  \caption{Comparison of the four methods---additive model (M1), spatial random forest (M2), spatial simultaneous autoregressive lag model (M3), and geographically weighted regression (M4)---along with six ensemble predictions---covariate-dependent stacking by the EM algorithm (CDSTE), covariate-dependent stacking by GAM (CDSTG), probabilistic stacking (PS), vanilla stacking (ST), simple averaging (SA), and smoothed-AIC (SAIC) in each spatial scenario. The top left-hand panel corresponds to Scenario 1, the top right-hand panel to Scenario 2, the bottom left-hand panel to Scenario 3, and the bottom right-hand panel to Scenario 4. \label{fig: comparison_sp}}
\end{figure}

\subsection*{Case 2: Weight Dependent on External Covariates}

Next, consider a setting where the data-generating process depends on space. We uniformly sample $400$ locations from $[-1,1]^2$. At each location, we generate a covariate as \eqref{eq: dgp_sp}. We prepare four different scenarios, where the mean structure $\mu$ in each scenario is defined by
\begin{itemize}
\item Scenario 1: $\mu = w + x_{3}^2 \exp\{-0.3(s_{1}^2 + s_{2}^2)\} + s_{2} \sin(2x_{2})$,
\item Scenario 2: $\mu = 2w + \frac12  \sin(\pi x_{1}x_{2}) +  (x_{3} - 0.5)^2 + \frac12 x_{4} + \frac14 x_{5}$,
\item Scenario 3: $\mu = 2w + (s_{1}+1)x_{1} + (1-s_{1})x_{3}^2$,
\item Scenario 4: $\mu = 2(s_{1}+1)w + x_{1} + (1-s_{1})x_{3}^2$,
\end{itemize}
where $w$ is the spatial random effect generated from a zero-mean Gaussian process with covariance kernel $(0.3)^2 \exp(-\|s-s'\|/0.3)$. Then, we sample response variable $y$ from $N(\mu, (0.7)^2)$.

In the experiment, we repeat the above procedure $100$ times, creating independent datasets each time. Each dataset consists of 400 points, of which $300$ points are allocated as training data and the remaining $100$ points as test data. To perform CDSTE and CDSTG, we prepare the following four models using all five covariates: (M1) additive model with thin plate splines, proposed by \cite{Hastie1986} and publicly available via the \texttt{R} package ``gam'' \citep{hastie2022gam}; (M2) spatial random forest with 50 trees as proposed by \cite{Saha2023random} and available by the \texttt{R} package ``RandomForestsGLS'' \citep{saha2020randomforestsgls}; (M3) spatial simultaneous autoregressive lag model~\citep{Anselin1988} available through \texttt{R} package ``spdep'' \citep{bivand2013spdep}; and (M4) geographically weighted regression \citep{Brunsdon1998}, performed by ``spgwr'' \citep{lu2014gwmodel}, where the bandwidth is selected by leave-one-out cross-validation. The other hyperparameters are maintained at their default settings. We perform the individual models (M1)--(M4) and utilize six combinatorial methods: CDSTE, CDSTG, PS, ST, SA, and SAIC, as in the previous experiment. To measure prediction performance, we use MSE computed on the test datasets.

Figure~\ref{fig: comparison_sp} presents comparisons of the four single methods and six ensemble predictions. The violin plots are based on $100$ datasets. First, we examine scenarios 1 and 2: CDSTE and CDSTG outperform PS, SAIC and SA. Unsurprisingly, it also performs better than single methods. However, CDSTs and ST are comparable, which means that there is little advantage to varying the weights spatially, as method M1 dominates the others almost everywhere. In fact, weight estimations by CDSTs do not vary in space (see the appendix for details). Next, we observe scenarios 3 and 4, where CDSTE and CDSTG dominate the other methods, including ST. In scenario 3, as the first coordinate of $s$ increases, the spatial linear regressivity becomes stronger, while the spatial nonparametric regressivity becomes stronger in scenario 4, whence the weights of M2 and M4 are proportional to the increase in the first coordinate of $s$, respectively (see the appendix for details). This result supports the predictive performance and explicative power of the CDST.

\subsection{Application to Land Price Data}\label{sec:app}
We finally demonstrate the CDST through a large-scale spatio-temporal prediction problem. To this end, we use the ``Real Estate Database 2018-2022'' provided by At Home Co., Ltd. and focus on predicting land price in four prefectures (Tokyo, Kanagawa, Chiba, and Saitama) in the Kanto region, Japan. The importance of accurate real estate price prediction is exemplified by its diverse applications, such as enabling local governments to develop more precise land use plans and allowing financial institutions to conduct more rigorous risk assessments for real estate-backed loans. The dataset contains land prices (yen), as well as auxiliary information on each land. We use the observations from October to December in 2022 as the test sample (8768 observations in total) and those from August 2018 to September 2022 as the training sample (186149 samples in total). We adopt seven covariates, land area (LA), floor area ratio (FAR), building coverage ratio (BCR), walking minutes from the nearest train station on foot (MF) and by bus (MB), and two dummy variables regarding zoning regulations (ZR1 and ZR2). For location information, the longitude and latitude information of each land, name of the nearest train station, and name of the ward are available. Time information is assigned using each quarter as a single time unit. As a result, the training period corresponds to quarters 1 through 17, while the test period corresponds to quarter 18.

To construct the prediction models for land prices, we consider the following three types of models:
\begin{itemize}
\item[-]
{\bf Station-level model:} \ \ The datasets are grouped according to the nearest train stations, and ordinary linear regression with five covariates (LA, FAR, BCR, MF, and MB) is applied to each grouped sample.

\item[-]
{\bf Ward-level model:} \ \ The datasets are grouped according to the wards, and a regression model with parametric linear effects of the seven covariates and the nonparametric effect of time information is applied to each grouped sample.

\item[-]
{\bf Prefecture-level model:} \ \ The datasets are grouped according to the prefectures and 
an additive model with nonparametric effects of the five continuous covariates (LA, FAR, BCR, MF and MB), longitude, latitude and time information, as well as with parametric effects of two dummy variables (ZR1 and ZR2) is applied to each grouped sample.
\end{itemize}
Since the sample size available for model estimation increases in the order of station-, ward-, and prefecture-level models, we vary model complexity (e.g., number of parameters) across the three types of models. Note that these three models are sufficiently simple and explicative for practical use, and the following integrations do not turn the prediction process into a ``black box.''

We combine the above three models through CDST with longitude, latitude, and time information as $\tilde{x}$ in \eqref{model: CDST} (i.e., model weight varies over space and time) and basis functions of the form $\phi_m(x)=\exp\{-\|x_1-\tilde{c}_{m1}\|_2^2/2h_1^2 - (x_2-\tilde{c}_{m2})^2/2h_2^2\}$, where $x=(x_1, x_2)$ with two-dimensional location information $x_1\in \R^2$ and time information $x_2\in \R$, $\tilde{c}_{m1}\in\R^2$ and $\tilde{c}_{m2}\in \R$ are the centers of basis functions and $h_1$ and $h_2$ are range parameters. To capture detailed spatial variations over more than 5000 locations and ensure smoothness over the limited 18 time periods, we apply $k$-means clustering with $20$ clusters for the spatial dimension and $5$ clusters for the temporal dimension, yielding a total of $M=100$ basis functions, and we determine $\tilde{c}_{m1}$ and $\tilde{c}_{m2}$ as a possible combination of cluster centers. Accordingly, we set the spatial range parameter to $h_1=0.2$ for fine spatial resolution, and the temporal range parameter to $h_2=2$ to enforce smoothness in time. For comparison, we also apply PS, ST and SA with equal weights. In CDSTE, CDSTG, PS and ST, we adopt 10-fold cross-validation to learn the model weight. To conduct a scale-independent evaluation of prediction accuracy, we employ the absolute percentage error (APE), defined as $100|y-\hat{g}(x)|/y$ for each of the test samples.

\begin{figure}[t]
  \begin{center}
  \includegraphics[width=0.95\linewidth]{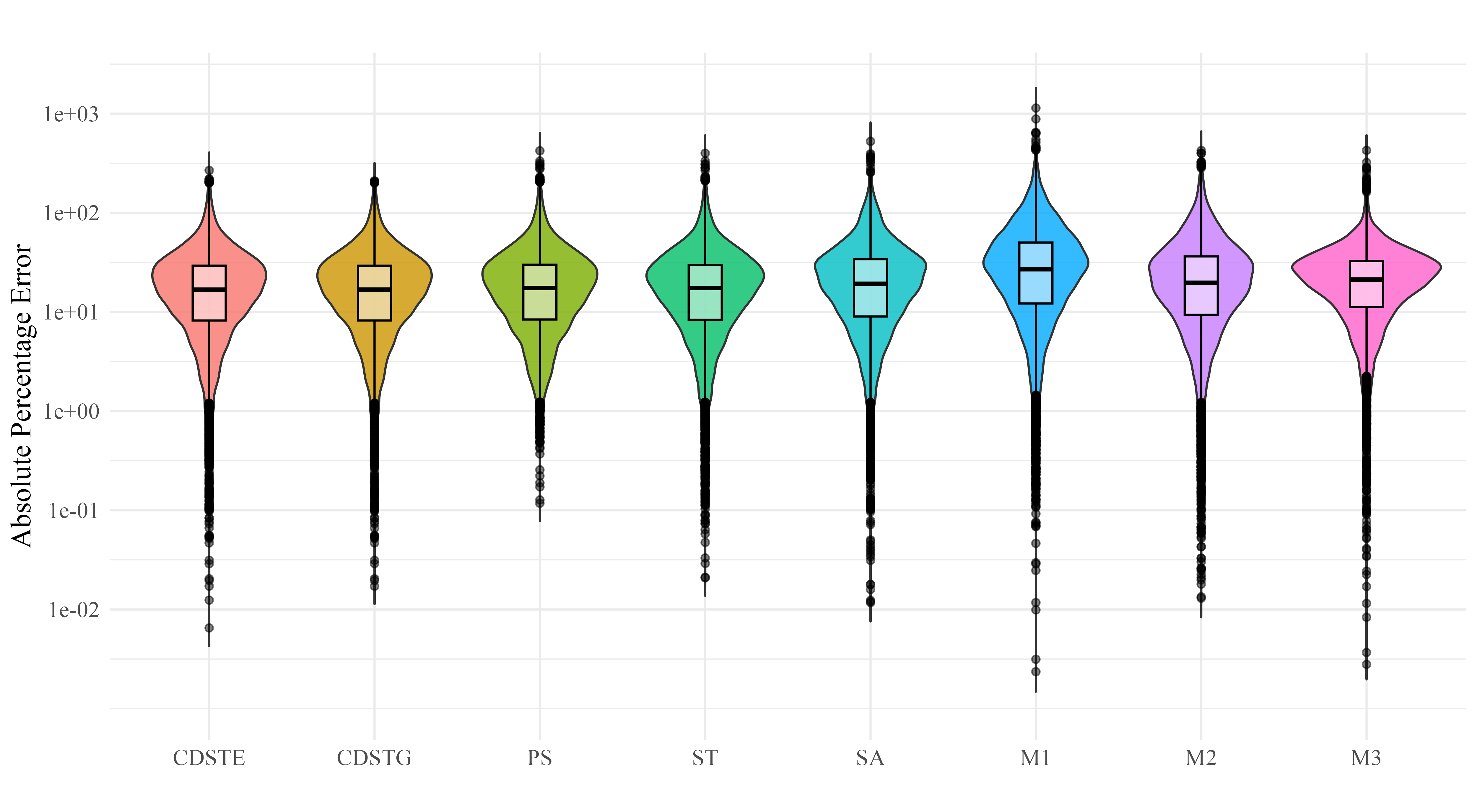}
  \end{center}
  \caption{Comparison of the three models: prefecture-level (M1), ward-level (M2), and station-level (M3) models, alongside five ensemble predictions: covariate-dependent stacking by the EM algorithm (CDSTE), covariate-dependent stacking by GAM (CDSTG), probabilistic stacking (PS), vanilla stacking (ST), and simple averaging (SA). \label{fig: comparison_app}}
\end{figure}

Figure~\ref{fig: comparison_app} comparatively analyzes the three single methods and their corresponding ensemble predictions. The distribution of APEs for all test samples across each method is represented through violin plots. Notably, the CDSTs exhibit fewer upper outliers, indicating a reduced frequency of significant prediction errors. The mean of the APEs for each model is as follows: 16.5 (CDSTE), 16.3 (CDSTG), 17.4 (PS), 17.2 (ST), 18.3 (SA), 34.9 (prefecture-level model), 23.8 (ward-level model), and 19.6 (station-level model). Subsequently, Figure~\ref{fig: weights_app} depicts the geographical distribution of weights assigned by CDSTE. The station-level model generally exhibits higher weights, particularly pronounced in Tokyo. This phenomenon may be attributed to Tokyo's complexity, characterized by its concentration of diverse types of districts and major stations. The ward-level model exerts an overall modest weight, potentially reflecting broader urban trends and temporal information. While the prefecture-level model generally has smaller weights, it can capture nuanced information, such as zoning patterns, which may not be addressed adequately by other models. The influence of this prefecture-level information is emphasized in some areas, such as Kanagawa.

\begin{figure}[t]
  \begin{center}
  \includegraphics[width=\linewidth]{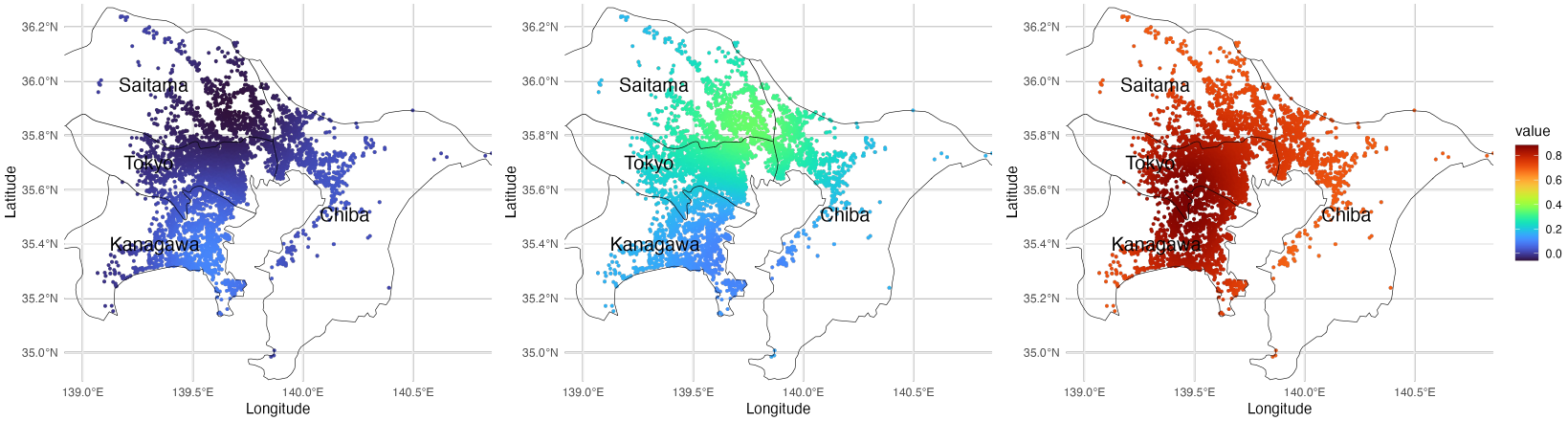}
  \end{center}
  \caption{Geographical distribution of weights assigned by the proposed method to the three models: prefecture-level (left), ward-level (center), and station-level (right) models.\label{fig: weights_app}}
\end{figure}

\subsection{Summary}
The numerical experiments and case study clearly demonstrate the advantages of our CDST, particularly in settings where the underlying data-generating mechanisms vary across different regions of the covariate space.
First, when relationships between base models and responses vary by location (spatial heterogeneity) or over time (temporal instability), the assignment of covariate‐dependent weights enables the final predictor to adapt to local variations. In both our simulations and the real estate case study, regions or time periods characterized by unique patterns were more accurately captured through locally adjusted weights, thereby enhancing predictive accuracy.
Second, the CDST proves particularly beneficial when the base models provide genuinely diverse perspectives (e.g., one model capturing smooth trends, another fitting local details, and yet another modeling temporal patterns). Conversely, if the base models are highly similar or if one model consistently dominates the others across all conditions, the advantages conferred by learning covariate‐dependent weights may be less pronounced.

\section{Concluding Remarks}\label{sec:conclusion}
This study proposes and evaluates a novel ensemble prediction method, CDST, which extends traditional stacking by allowing model weights to vary as a function of covariates of interest. Our theoretical analysis thus provides a solid foundation for understanding the method's behavior for large samples, while our empirical analysis demonstrates its practical utility and interpretability.
The simulation studies revealed that CDST consistently outperforms other ensemble methods, particularly in scenarios where the true data-generating process varies across the covariate space. This advantage is further corroborated by our application to large-scale land price prediction, where CDST exhibited superior predictive performance compared to both individual models and other ensemble methods.

Our findings have significant implications for researchers and practitioners dealing with intricate prediction tasks, especially in spatio-temporal contexts. The adaptability of CDST in accommodating heterogeneous relationships across the covariate space renders it a potent tool for enhancing prediction accuracy in diverse fields. Future research directions may include extending CDST to other problems, such as heterogeneous treatment effects, devising more computationally efficient approaches for exceptionally high-dimensional problems, and exploring its applicability in other domains, including climate modeling and financial forecasting.

\section*{Computer Programs}
Computer programs used for our numerical experiments in Section~\ref{sec:num} were developed for execution in the \texttt{R} statistical computing environment. The programs are available in the GitHub repository \href{https://github.com/TomWaka/CovariateDependentStacking}{https://github.com/TomWaka/CovariateDependentStacking}.

\section*{Acknowledgements}
This research was the result of the joint research with the Center for Spatial Information Science, the University of Tokyo (No. 1342) and used the following data: ``Real Estate Database 2018-2022'' provided by At Home Co., Ltd. of Japan.

\section*{Funding}
This work was supported by the Japan Society for the Promotion of Science (JSPS KAKENHI) grant numbers, 22KJ1041, 20H00080 and 21H00699 and Japan Science and Technology Agency (JST, ACT-X) grant number JPMJAX23CS.


\section*{Appendix A: Discussion on Oracle Inequality and Penalty}
First, we introduce the oracle inequality under a general loss function for estimates obtained through standard cross-validation. Second, we present the oracle inequality for estimators derived from penalized cross-validation. Finally, we discuss the advantages of the latter approach.

\subsection{Oracle Inequality for General Loss without Penalty}
Let $(S, \mS, P)$ be a probability space and $(X_1, Y_1),\ldots,(X_n, Y_n)$ be a set of independent and identically distributed random elements from $P$, denoted by $\mD$. We consider the following model:
\begin{equation*}
    Y = g_0(X) + \varepsilon,
\end{equation*}
where $g_0$ is a true function and $\varepsilon$ is a random error term. The prediction of $y$ given the new $x$ is of interest. To quantify the goodness of predictor $g(x)$, we set measurable loss function $L:\R\times\R\to \R$ and risk function $\E_{(X,Y)\sim P} [L(g(X), Y)]$. 

We randomly split dataset $\mD= \mD_0 + \mD_1$ by providing $\{0,1\}$-valued random variable (indicator) to each pair $(X_i,Y_i)$ and assigning the pair that has realized indicator $j$ to $\mD_j$ for $j=0$ and $1$. $\mD_0$ is used as training data and $\mD_1$ as test (validation) data. Let $P_I$ and $\E_{I}$ be the probability distribution and expectation operator over the random split assignments $I = (I_1, \ldots, I_n)$ (independent of $\mD$), $\mathbb{P}_j$ and $\E_{\mathbb{P}_j}$ be empirical measures respectively corresponding to $\mD_j$ and its expectation operator for $j=0$ and $1$, and $n_j$ be the cardinality of $\mD_j$. Assume $n_1 = c n$ with $0<c<1$. Suppose we have a countable set of measurable functions $\mG := \{ g_{\Theta}(x)=\sum_{j=1}^J w_{j,\Theta}(x)\fh_{j}(x) \mid \Theta \in \R^{JM+J} \}$, where $\fh_{j}(\cdot) \ (j=1,\ldots,J)$\ be a given base predictor and $w_{j,\Theta}(x)$ is a weight function parametrized by $\Theta$. Remark that $\Theta (= \Theta(\mD_0))$ depends on $\mD_0$ since we learn it in the training process. 

Using the test data $\mathcal{D}_1$, we define stacking predictor $\hat{g}$ as the minimizer of split-averaged empirical risk: 
\begin{equation*}
    \hat{g} = g_{\hat{\Theta}}(x;\mD_0,\mD_1) :=  \argmin_{g_{\Theta}\in\mG}  \E_I\Biggl[\frac{1}{n_1}\sum_{(X_i,Y_i)\in \mD_1} L\bigl(g_{\Theta(\mD_0)}(X_i),Y_i\bigr)\Biggr] .
\end{equation*}
Additionally, we denote the oracle predictor in $\mG$ by
\begin{align*}
    g^* &:= \argmin_{g_{\Theta}\in\mG} \E_{I} \left [ \E_{(X,Y)\sim P} \left[L(g_{\Theta(\mD_0)}(X),Y)\right] \right] \\
    & = \argmin_{g_{\Theta}\in\mG}  \E_I \left [ \E_{\mD_1\sim P^{\otimes n_1}} \left[\frac{1}{n_1}\sum_{(X_i,Y_i)\in \mD_1} L\bigl(g_{\Theta(\mD_0)}(X_i),Y_i\bigr) \right] \right],
\end{align*}
where $P^{\otimes n_1}$ represents the measure obtained by taking the Cartesian product of $n_1$ copies of the probability distribution $P$.

Before proceeding to the main theorem, we define bivariate function set $\mL:= \{ L(g(\cdot),\cdot): S \to \R \mid  g\in \mG \}$ and a Bernstein pair of a function. For a given measurable function $h:S\to\R$, $(M(h),v(h))\in \R^2$ is called a Bernstein pair if it satisfies
\begin{align*}
    M(h)^2 \E_P \left[\exp\left( \frac{|h|}{M(h)} \right) - 1 - \frac{|h|}{M(h)} \right] \le \frac{1}{2} v(h).
\end{align*} 
The existence of a Bernstein pair for a given function is equivalent to the moment condition of Bernstein's inequality \citep[refer to Chapter 2.2.4 in][]{van2023weak} and is also a weaker condition than the boundedness of the $L_{\infty}$-norm. For more details about Bernstein numbers, see \citet{vaart2006oracle}.

\begin{theorem}[Oracle Inequality for General Loss without Penalty]\label{thm:general}
For any $\delta>0$, $\alpha>0$ and $1\le p\le 2$, the following inequality holds:
\begin{align*}
    & \E_{I} \left[ \E_{(X,Y)\sim P} \left[L(\hat{g}(X),Y)\right]\right] \\
    & \le (1+2\alpha)\E_{I} \left[ \E_{(X,Y)\sim  P } \left[L(g^*(X),Y)\right]\right]
     + (4+6\alpha)\delta\\
    &\quad + (1+\alpha) \frac{16}{n_1^{1/p}} \log\left\{1+\mN(\delta, \mL,\|\cdot\|_{\infty})\right\} \sup_{\ell\in\mL} 
        \left\{ \frac{M(\ell)}{n_1^{1-1/p}} + \left(\frac{v(\ell)}{(\E_P [\ell])^{2-p}}\right)^{1/p}\left( \frac{1+\alpha}{\alpha}\right)^{2/p-1} \right\} .
\end{align*}
\end{theorem}

\subsection{Oracle Inequality for General Loss with Penalty}
We consider the same setting as in the previous section but add a measurable penalty function $Q_{\lambda}(\Theta)$ with scale parameter $\lambda$ to the objective function and consider the following estimator:
\begin{equation*}
\hat{g} := \argmin_{g_{\Theta}\in\mG}  \E_I\Biggl[\frac{1}{n_1}\sum_{(X_i,Y_i)\in \mD_1} L\bigl(g_{\Theta(\mD_0)}(X_i),Y_i\bigr) + \frac{1}{n} Q_{\lambda}(\Theta) \Biggr] .
\end{equation*}
We define the penalized oracle predictor in $\mG$ as 
\begin{align*}
    g^* &:= \argmin_{g_{\Theta}\in\mG} \E_{I} \left [ \E_{(X,Y)\sim P} \left[L(g_{\Theta(\mD_0)}(X),Y)\right]  + Q_{\lambda}(\Theta) \right] \\
    & = \argmin_{g_{\Theta}\in\mG}  \E_I \left [ \E_{\mD_1\sim P^{\otimes n_1}} \left[\frac{1}{n_1}\sum_{(X_i,Y_i)\in \mD_1} L\bigl(g_{\Theta(\mD_0)}(X_i),Y_i\bigr) \right] + \frac1n Q_{\lambda}(\Theta) \right]
\end{align*}
and evaluate how $\hat{g}$ can approximate $g^*$.

\begin{theorem}[Oracle Inequality for General Loss with Penalty]\label{thm:penalty}
For any $\delta>0$, $\alpha>0$, $0< p\le 1$ and $0< q \le 1$, the following inequality holds:
\begin{align*}
    & \E_{I} \left[ \E_{(X,Y)\sim P} \left[L(\hat{g}(X),Y)\right]\right] \\
    & \le (1+2\alpha)\E_{I} \left[ \E_{ (X,Y)\sim P } \left[L(g^*(X),Y)\right] +\frac1n Q_{\lambda}(\Theta^*)  \right]
     + (4+6\alpha)\delta\\
    &\quad + (1+\alpha) \frac{16}{n_1} \left\{\log\left(1+D_q+\mN(\delta, \mL,\|\cdot\|_{\infty})\right)\right\}^{1/q} \sup_{\ell\in\mL} 
        \left(   \frac{{(1+\alpha)}^{1-q}  M(\ell)}{C_q {\alpha}^{1-q}  Q_{\lambda}(\Theta)^{1-q}} \right)^{1/q} \\
    &\quad + (1+\alpha) \frac{16}{n_1} \left\{\log\left(1+D_p+\mN(\delta, \mL,\|\cdot\|_{\infty})\right)\right\}^{1/p} \sup_{\ell\in\mL} 
        \left( \frac{ {(1+\alpha)}^{2-p} v(\ell)}{C_p  {\alpha}^{2-p}  \E_P [\ell] Q_{\lambda}(\Theta)^{1-p}} \right)^{1/p},
\end{align*}
where $C_p>0$ and $D_p\ge 0$ are constants equal to $1$ and $0$ for $p=1$. Here, $Q_{\lambda}(\Theta)$ is a variable depending on $\ell\in \mL$.
\end{theorem}

One of the principal distinctions between Theorem~\ref{thm:general}, which lacks a penalty term, and the aforementioned theorem is the presence of $\frac{1}{n}Q_{\lambda}(\Theta^*)$. This discrepancy diminishes as sample size increases. Another notable difference is the division of $M(\ell)$ and $v(\ell)/\E_P [\ell]$ by $Q_{\lambda}(\Theta)^{1-p}$. For squared loss (or absolute loss) functions, $M(\ell)$ and $v(\ell)/\E_P[\ell]$ are approximately equivalent to $\|g-g^*\|_{\infty}^2$ (or $\|g-g^*\|_{\infty}$, respectively). Under these circumstances, $M(\ell)$ and $v(\ell)/\E_P [\ell]$ may potentially diverge. However, by implementing regularization, for instance, by setting $Q_{\lambda}(\Theta) \approx \|g\|_{\infty}^{1/(1-p)}$, it is possible to avoid divergence, provided that $\|g\|_{\infty}$ remains bounded ($f_1 \approx f_2$ means that the ratio of two functions, $f_1/f_2$, is bounded above and below by positive constants). Although configuring the penalty function in this manner is practically unusual due to computation, this explains the benefits of regularization.

\section*{Appendix B: Proofs}

This section presents the technical proofs.

\begin{proof}[Proof of Theorem~\ref{thm:qdrt}]
    In Theorem~\ref{thm:penalty}, take $p=q=1$ and $\alpha = \delta = n_1^{-1/2}$. The covering number of $\mL$ with respect to the uniform norm is $\mN(\delta, \mL,\|\cdot\|_{\infty})\lesssim \delta^{-d-1}$. A Bernstein pair of $\ell\in \mL$ is $v(\ell)= 2\E_P[(g-g_0)^2] \left(e B+ 8t^{-2}\left\|\E_{P}\left[e^{t|\varepsilon|}\given X\right] \right\|_{\infty} \right)$ and $M(\ell)= 4 \max\{t^{-1},1\} \max \{B,1\}$, where $g$ corresponds to $\ell$.
\end{proof}

\begin{proof}[Proof of Theorem~\ref{thm:general}]
    We first consider the following decomposition for any $\alpha>0$:
    \begin{align*}
        &\int \int L(\hat{g}(X),Y) \mathrm{d}P \mathrm{d} P_I\\
        & \leq \int \int  L(g^*(X),Y) \mathrm{d}(1+\alpha)\mathbb{P}_1 \mathrm{d} P_I- \int \int L(\hat{g}(X),Y) \mathrm{d}\{(1+\alpha)\mathbb{P}_1 - P\}\mathrm{d} P_I \\
        & = \int \int  L(g^*(X),Y) \mathrm{d}(1+2\alpha)P\mathrm{d} P_I + \int \int L(g^*(X),Y)\mathrm{d}\{(1+\alpha)\mathbb{P}_1 - (1+2\alpha)P\}\mathrm{d} P_I\\
        &\quad + \int \int -L(\hat{g}(X),Y)  \mathrm{d} \{(1+\alpha)\mathbb{P}_1 - P\}\mathrm{d} P_I \\
        & \leq  (1+2\alpha) \E_I \left[\E_{ (X,Y)\sim P } \left[L(g^*(X),Y)\right]\right] + 
        \E_{I}\left[ \sup_{g\in\mG} \int L(g(X),Y) \mathrm{d}\{(1+\alpha)(\mathbb{P}_1 - P)-\alpha P\}\right]  \\
        &\quad + \E_{I}\left[\sup_{g\in\mG} \int -L(g(X),Y)\mathrm{d} \{(1+\alpha)(\mathbb{P}_1 - P)+\alpha P \}\right].
    \end{align*}
    
    To obtain the maximal inequalities of the second and third terms on the right-hand side, we utilize the chaining argument \citep{vershynin2018high,gine2021mathematical}. Because the same discussion and upper bound are valid for both the third and second terms, we focus solely on the second term. First, we fix one of the fewest $\delta$-net, denoted by $\mL_{net}$, of countable measurable function set $\mL:= \{ L(g(\cdot),\cdot): X \to \R \mid  g\in \mG \}$ with respect to $L_{\infty}$-distance. Then, we have
    \begin{align}
        &\E_{I}\left[\sup_{g\in\mG} \int L(g(X),Y) \mathrm{d}\{(1+\alpha)(\mathbb{P}_1 - P)-\alpha P\} \right] \nonumber \\
        &\le \E_{I}\left[\sup_{\ell,\ell'\in\mL: \| \ell-\ell'\|_{\infty}<\delta} \int \left\{\ell(X,Y)-\ell'(X,Y)\right\}\mathrm{d}\{(1+\alpha)(\mathbb{P}_1 - P)-\alpha P\}\right] \nonumber \\
        &\quad + \E_{I}\left[\sup_{\ell\in\mL_{net}} \int \ell(X,Y)\mathrm{d}\{(1+\alpha)(\mathbb{P}_1 - P)-\alpha P\} \right]. \label{ineq:chain}
    \end{align}
    The first term on the right-hand side is the supremum of the variation within the $\delta$-neighborhoods. For any pair $\ell,\ell'\in\mL$ satisfying $\| \ell-\ell'\|_{\infty}<\delta$, we have
    \begin{align}
            &\int \left\{\ell(X,Y)-\ell'(X,Y)\right\}\mathrm{d}\{(1+\alpha)(\mathbb{P}_1 - P)-\alpha P\} \nonumber\\
            &\leq (1+\alpha) \int \left|\ell(X,Y)-\ell'(X,Y)\right|\mathrm{d}\mathbb{P}_1  + (1+2\alpha) \int \left|\ell(X,Y)-\ell'(X,Y)\right|\mathrm{d}P \nonumber \\
            &\le (2+3\alpha)\delta. \label{eq: delta-neighbor}
    \end{align}
    The second term is the supremum over $\delta$-net and can be bounded from above by Lemma~2.2 in \cite{vaart2006oracle}: 
    \begin{align}
        &(1+\alpha) \E_{I}\left[ \sup_{\ell\in\mL_{net}} \int  \left[\ell(X,Y)\right] \mathrm{d}\{ (\mathbb{P}_1 - P)-\alpha' P \}\right] \nonumber\\
        &\leq (1+\alpha) \frac{8}{n_1^{1/p}} \log\left(1+\mN(\delta, \mL,\|\cdot\|_{\infty})\right) \sup_{\ell\in\mL} 
        \left\{ \frac{M(\ell)}{n_1^{1-1/p}} + \left(\frac{v(\ell)}{(\alpha'\E_P [\ell])^{2-p}}\right)^{1/p} \right\}, \label{eq: delta-net}
    \end{align}
    where $\alpha' = \alpha/(1+\alpha)$. Substituting \eqref{eq: delta-neighbor}~and~\eqref{eq: delta-net} into \eqref{ineq:chain} concludes the proof.
\end{proof}

\begin{proof}[Proof of Theorem~\ref{thm:penalty}]
The basic flow is similar to the Proof of Theorem~\ref{thm:general}, but we add the penalty term to $L(g(\cdot),\cdot)$ and follow the same argument. Note that Lemma~3.1 in \cite{vaart2006oracle} and simple algebra provide the bound corresponding to the second term in~\eqref{ineq:chain}.
\end{proof}

\section*{Appendix C: EM algorithm}

In this section, we derive the EM algorithm, which is also employed in our numerical experiments. This discussion continues from Section 2.2 of the main text.

To obtain $\Thh$ as well as the tuning parameters, we consider the following working model for $y_i\in\R$:
\begin{equation}\label{working-model}
y_i|\gamma_1,\ldots,\gamma_J\sim N\Big(\sum_{j=1}^J \big\{\mu_j + E(x_i)^\top \gamma_j\big\}\fh_{j, -i}(x_i), \sigma^2\Big), \qquad \gamma_j\sim N(0, \tau_j^2I_M), 
\end{equation}
where $\sigma^2\in\R_+$ is an additional (nuisance) parameter.
Note that, under working model \eqref{working-model}, the conditional posterior mode of $\Theta$ given $\{\gamma_j\}$ is equivalent to \eqref{optimal-weight} with $\lambda_j=\sigma^2/\tau_j^2$.
We then employ an EM algorithm \citep{dempster1977maximum} to simultaneously estimate $\{\mu_j\}$, $\{\tau_j\}$ and $\sigma^2$ by considering $\gamma_j$ as unobserved random effects.  

The conditional distribution of $\gamma\in \R^{JM}$ given the observed data is $N(m_\gamma, S_\gamma)$ with
\begin{equation}\label{E-step}
m_\gamma=\frac{1}{\sigma^2} S_\gamma W^{\top} \bigg(y-\sum_{j=1}^J \mu_j F_j\bigg), \ \ \ \ 
S_\gamma=\bigg(\frac{W^\top W}{\sigma^2} + D\otimes I_M\bigg)^{-1},
\end{equation}
where $y=(y_1,\ldots,y_n)^{\top}$, $F_j=(\fh_{j, -1}(x_1),\ldots,\fh_{j, -n}(x_n))^\top$, $D=\mathrm{diag}(1/\tau_1^2,\ldots, 1/\tau_J^2)\in \R^{J\times J}$, and $W$ is an $n\times JM$ matrix whose $i$th row is $(E(x_i)^\top\fh_{1, -i}(x_i),\ldots, E(x_i)^\top\fh_{J, -i}(x_i))$.
Then, the expectation of the complete log-likelihood to be maximized in the \textit{M-step} can be expressed as
\begin{align*}
Q_{\lambda}(\Psi)=-&n\log\sigma -\frac{1}{2\sigma^2}\sum_{i=1}^n \bigg(y_i^{\ast} -\sum_{j=1}^J \mu_j \fh_{j,-i}\bigg)^2   \\
       &  - M\sum_{j=1}^J\log \tau_j - \sum_{j=1}^J \frac{m_{\gamma (j)}^\top m_{\gamma (j)}+{\rm tr}(S_{\gamma(jj)}) }{2\tau_j^2},
\end{align*}
where $y_i^{\ast}=y_i-\sum_{j=1}^JE(x_i)^\top m_{\gamma(j)}\fh_{j,-i}$, $m_{\gamma(j)}$ and $S_{\gamma(jj)}$ are respectively the mean vector and covariance matrix of the conditional distribution of $\gamma_j$ given the observed data; and $\Psi$ is a collection of unknown parameters, $\{\mu_j\}$, $\{\tau_j\}$ and $\sigma^2$. Then, the maximization steps of $Q_{\lambda}(\Psi)$ are obtained in closed form as shown in Algorithm~\ref{algo:EM}.

\begin{algo}\label{algo:EM}
With initial values, $\Psi_{(0)}=(\mu_{1(0)},\ldots,\mu_{J(0)}, \tau_{1(0)},\ldots,\tau_{J(0)}, \sigma^2_{(0)})$ and $r=0$, repeat the following steps until convergence: 
\begin{enumerate}
\item
(E-step) \ \ Using $\Psi_{(r)}$, compute the posterior expectation and covariance matrix of $\gamma_j$ by \eqref{E-step}. 

\item 
(M-step) \ \ Update $\Psi$ as follows: 
\begin{align*}
(\mu_{1(r+1)},\ldots,\mu_{J(r+1)}) \ &\leftarrow \ \left(\sum_{i=1}^n F_iF_i^{\top}\right)^{-1}\sum_{i=1}^n y_i^{\ast}F_i,\\
\sigma^2_{(r+1)} \  &\leftarrow \ \frac1n\sum_{i=1}^n\bigg(y_i^{\ast} -\sum_{j=1}^J \mu_{j(r+1)} \fh_{j,-i}\bigg)^2,\\
\tau_{j(r+1)}^2 \  &\leftarrow \ \frac1M\left\{m_{\gamma (j)}^\top m_{\gamma (j)}+{\rm tr}(S_{\gamma(jj)})\right\}, \ \ \ j=1,\ldots,J.
\end{align*}
\end{enumerate}
\end{algo}

Given parameter estimate $\widehat{\Psi}$, the plug-in posterior distribution of $\gamma_j$ can be obtained by \eqref{E-step}. In particular, the point estimates of $\gamma_j$ can be obtained as plug-in posterior expectation $m_{\gamma}$ in \eqref{E-step} with the parameters replaced by their estimates. Then, the estimated model weight is given by $\widehat{w}_j(x)=\widehat{\mu}_j + E(x)^\top \widehat{\gamma}_j$, which achieves a flexible ensemble predictor of form \eqref{model: CDST} on arbitrary point $x$.

\section*{Appendix D: Additional Results of the Numerical Study}

This section presents the supplementary results not included in Section~4.

\begin{figure}[tb]
  \begin{center}
  \includegraphics[width=0.9\linewidth]{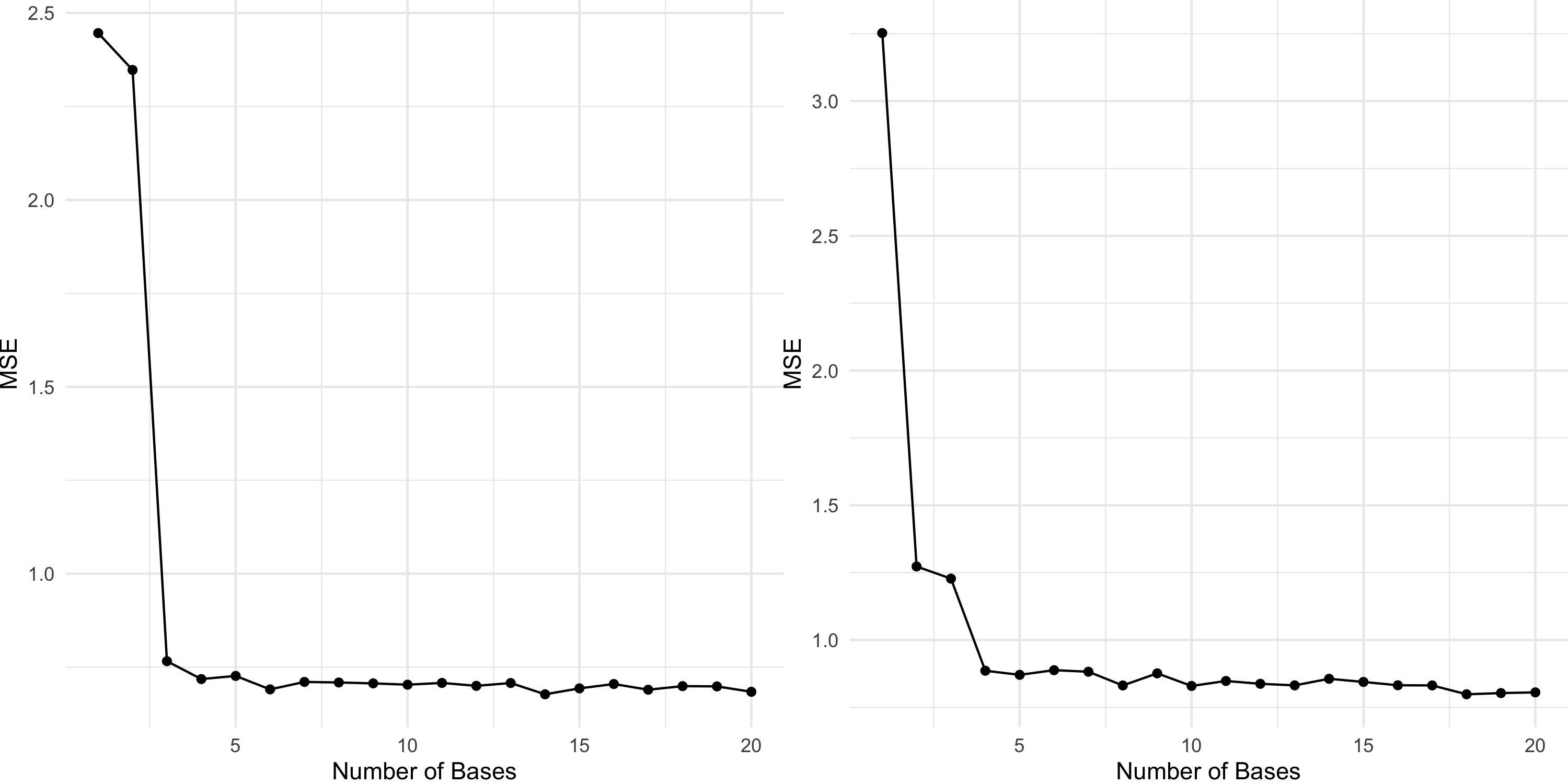}
  \end{center}
  \caption{ Relationship between the number of bases and the mean squared error (MSE) for Case 1 (left panel) and for Case 2 (right panel) in Section 4.1.
  \label{fig: number_bases}}
\end{figure}

Figure~\ref{fig: number_bases} illustrates the relationship between the number of basis functions and the mean squared error (MSE) on the test data for Cases~1~and~2 in Section~4.1. In this setting, the decrease in MSE nearly saturates when using around 3--4 basis functions, and further increases have little impact on the results. Therefore, employing a sufficiently large number of basis functions (approximately 10) ensures consistency.

Figure~\ref{fig: comparison_cov_weights} shows how CDST (via the EM algorithm) assigns weights to the linear regression (M1), additive model (M2), random forest (M3), and Gaussian process regression (M4) in Case 1 (weight dependent on internal covariates) in Section~\ref{subsec: comparison}. The left-hand side of the figure corresponds to Scenario 1 and the right-hand side to Scenario 2. In Scenario 1, we observe that each model is assigned constant (covariate-independent) weights, which is consistent with conventional stacking methods. As noted in the main text, there is no significant difference in accuracy between CDST and conventional stacking in this scenario. Conversely, Scenario 2 demonstrates that the weights of each model vary across the $(X_1,X_2)$ space. Notably, the additive model's weight increases in the regions where $X_1$ is small. This is a reasonable result, considering that the corresponding areas exhibit a nonlinear structure ($x_3^2$).

\begin{figure}[tb]
  \begin{center}
  \includegraphics[width=0.9\linewidth]{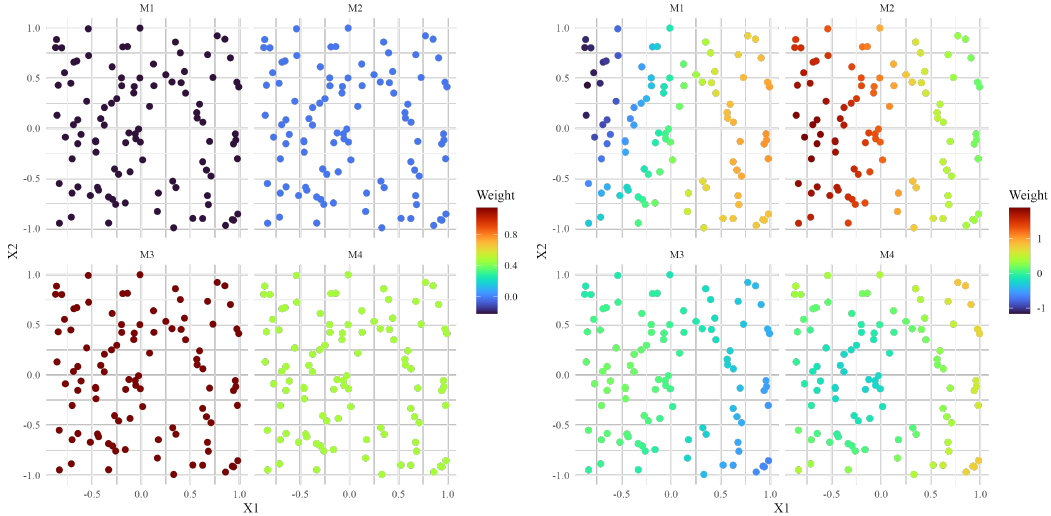}
  \end{center}
  \caption{Weights assigned by the proposed method for the four methods, linear regression (M1), additive model (M2), random forest (M3), and Gaussian process regression (M4), in the covariate-dependent data generating scenario. The left-hand set corresponds to Scenario 1 and the right-hand set to Scenario 2.\label{fig: comparison_cov_weights}}
\end{figure}

Figure~\ref{fig: comparison_sp_weights} depicts how CDST (via the EM algorithm) allocates weights to the additive model (M1), spatial random forest (M2), spatial simultaneous autoregressive lag model (M3), and geographically weighted regression (M4) in Case 2 (weight dependent on external covariates) from Section~\ref{subsec: comparison}. The figure is divided into four quadrants: top-left for Scenario 1, top-right for Scenario 2, bottom-left for Scenario 3, and bottom-right for Scenario 4. For Scenarios 1 and 2, each model is assigned constant (spatially independent) weights, similar to conventional stacking methods. Consequently, there is no significant difference in accuracy between CDST and conventional stacking in these scenarios. However, Scenarios 3 and 4 demonstrate that the weights for each model vary according to latitude and longitude. In both scenarios, the additive model (M1) receives increased weight in the western regions, which is consistent with the quadratic structure present in these areas. By contrast, the weight of the spatial random forest (M2) increases toward the eastern regions in Scenario 4, aligned with the increase in nonparametric effects (spatial random effects) in the eastern regions.

These results underscore the adaptive nature of the CDST in allocating model weights based on spatial and covariate-dependent factors, thereby capturing complex underlying structures in the data.

\begin{figure}[tb]
  \begin{center}
  \includegraphics[width=0.9\linewidth]{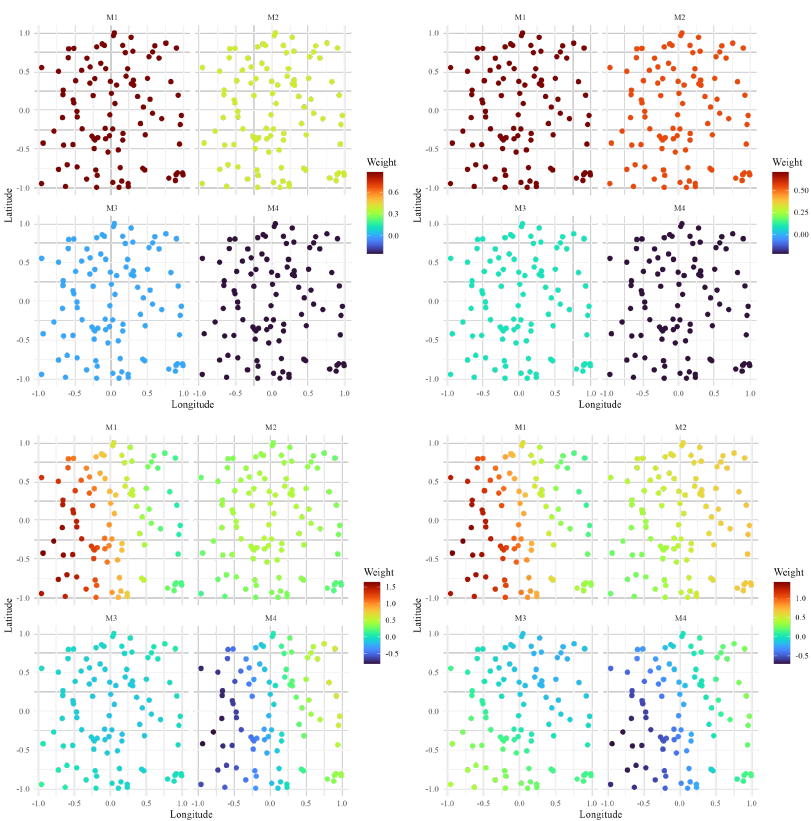}
  \end{center}
  \caption{Weights assigned by the proposed method for the four methods, additive model (M1), spatial random forest (M2), spatial simultaneous autoregressive lag model (M3), and geographically weighted regression (M4), in spatial settings. The top left-hand set corresponds to Scenario 1, the top right-hand set to Scenario 2, the bottom left-hand set to Scenario 3, and the bottom right-hand set to Scenario 4.\label{fig: comparison_sp_weights}}
\end{figure}

\bibliographystyle{chicago}
\bibliography{ref}

\end{document}